\documentclass[a4paper,11pt]{llncs}
\usepackage[latin1]{inputenc}
\usepackage[T1]{fontenc}
\usepackage[english]{babel}
 \usepackage{amsmath}
 \usepackage{amstext}
 \usepackage{amssymb}
 \usepackage{amsfonts}
 \usepackage{algpseudocode}
\usepackage{algorithmicx}
\usepackage{algorithm}
\usepackage{float}
\usepackage[margin=2cm]{geometry}
\usepackage{tikz}
  \usetikzlibrary{arrows,automata,shapes}
\usepackage{hyperref}

\tikzstyle{every picture}=[
  >=stealth', shorten >=1pt, node distance=1.44cm,auto,bend angle=45,initial text=,
  every state/.style={inner sep=0.75mm, minimum size=1mm},font=\scriptsize,
]

\title{Construction of rational expression from tree automata using a generalization of Arden's Lemma}

\author{
  Younes Guellouma\inst{1,3} 
  \and Ludovic Mignot\inst{2}
  \and Hadda Cherroun\inst{1,3} 
  \and Djelloul Ziadi\inst{2,3}  
}

\institute{
  Laboratoire LIM, Universit\'{e} Amar Telidji, Laghouat, Alg\'{e}rie\\
  \email{\{y.guellouma,hadda\_cherroun\}@mail.lagh-univ.dz}
  \and LITIS, Universit\'e de Rouen, 76801 Saint-\'Etienne du Rouvray Cedex, France\\
   \email{\{ludovic.mignot,djelloul.ziadi\}@univ-rouen.fr}
 \and supported by the MESRS - Algeria under Project 8/U03/7015.
}

\begin{document}

\maketitle

\begin{abstract}

 Arden's Lemma is a classical result in language theory
 allowing the computation of a rational expression denoting the language recognized by a finite string automaton.
 In this paper we generalize this important lemma to the rational tree languages. 
 Moreover,
 we propose also a construction of a rational tree expression which denotes the accepted tree language of a finite tree automaton.

\end{abstract}

\begin{keywords}
    Tree automata theory,
    Arden's lemma,
    Rational expression.

\end{keywords}

\section{Introduction}

Trees are natural 
structures
used in many fields in computer sciences like XML~\cite{XML1}, indexing,  natural language processing, code generation for compilers, term rewriting~\cite{tata2007}, cryptography~\cite{DBLP:conf/stacs/2004} etc. This large use of this structure leads to concider the theoretical basics of a such notion.

In fact,  in many cases, the problem of trees blow-up causes difficulties of storage and representation of this large amount of data. To outcome this problem, many solutions persist. Among them, the use of tree automata and rational tree expressions as compact and finite structures that recognize and represent infinite tree sets.

As a part of the formal language theory, trees are considered as 
a
generalization of strings. Indeed in the late of 1960s~\cite{Brain97,magidor1969finite}, many researches generalize strings to  trees and many notions 
appeared like tree languages, tree automata, rational tree expressions, tree 
grammars,
etc.

Since tree automata are beneficial in an acceptance point of view and the rational expressions in a descriptive one, an equivalence between the two representations must be resolved. Fortunately, Kleene result~\cite{TH68} states this equivalence between the accepted language of  tree automata  and the language denoted by rational expressions. 

Kleene theorem proves that the set of languages denoted by all rational expressions over the ranked alphabet $\Sigma$ noted $Rat(\Sigma)$ and the set of all recognized languages over $\Sigma$ noted $Rec(\Sigma)$ are equivalent. This can be checked also  by   verifying   the two inclusions $Rat(\Sigma)\subseteq Rec(\Sigma)$ and $Rec(\Sigma)\subseteq Rat(\Sigma')$ where $\Sigma\subseteq\Sigma'$.
In other words, 
any tree language is recognized by some automaton if and only if it is denoted by some rational expression.
Thus two constructions can be pulled up.

From a rational expression to tree automata, several techniques exist. First, Kuske et 
Meinecke
~\cite{DBLP:journals/ita/KuskeM11} generalize the notion of languages 
partial
derivation~\cite{DBLP:journals/tcs/Antimirov96} from strings to trees and propose a tree equation  automaton which is constructed from a derivation of a linearized version of rational expressions. They use the ZPC structure~\cite{DBLP:journals/ijac/ChamparnaudZ01} to reach best complexity. After that,  Mignot et al.~\cite{DBLP:journals/corr/MignotSZ14} propose an efficient algorithm to compute this generalized tree equation automata. Next, Laugerotte et al.~\cite{DBLP:conf/lata/LaugerotteSZ13} generalize position automata to trees.
Finally, the morphic links between these constructions have been defined in~\cite{AFL2014}.

In this paper, we propose a construction of  the second way of Kleene Theorem, the passage from a tree automaton to its rational tree expression. For this reason we propose a generalization of Arden's 
Lemma
for strings to trees. 
The
complexity of a such construction is exponential.

Section \ref{pre} recalls some preliminaries and basic properties. 
We generalize the notion of equation system in Section~\ref{sec eq syst}.
Next the generalization of Arden's lemma to trees and its proof is given in Section \ref{ar}, leading to the computation of some solutions for particular recursive systems. 
Finally, we show how to compute a rational expression denoting the language recognized by a tree automaton in Section~\ref{co}.

\section{Preliminaries and Basic Properties}\label{pre} 

Let $\Sigma=\bigcup_{n\geq 0} \Sigma_n$ be a graded alphabet.
A \emph{tree} $t$ over $\Sigma$ is inductively defined by $t=f(t_1,\ldots,t_n)$ with $f\in\Sigma_n$ and $t_1,\ldots,t_n$ any $n$ trees over $\Sigma$.
A \emph{tree language} is a subset of $T(\Sigma)$. 
The \emph{subtrees set} $\mathrm{St}(t)$ of a tree $t=f(t_1,\ldots,t_n)$ is defined by $\mathrm{St}(t)=\{t\}\cup \bigcup ^n_{k=1} \mathrm{St}(t_k)$. 
This set is extended to tree languages, and the \emph{subtrees set} $\mathrm{St}(L)$ of a tree language $L\subset T(\Sigma)$ is $\mathrm{St}(L)=\bigcup_{t\in L} \mathrm{St}(t)$.
The \emph{height} of a tree $t$ in $T(\Sigma)$ is defined inductively by $\mathrm{Height}(f(t_1,\ldots,t_n))=1+\max\{\mathrm{Height}(t_i)\mid 1\leq i\leq n\}$ where $f$ is a symbol in $\Sigma_n$ and $t_1,\ldots,t_n$ are any $n$ trees over $\Sigma$.

A \emph{finite tree automaton} (FTA) over $\Sigma$ is a $4$-tuple $\mathcal{A}=(\Sigma,Q,Q_f,\Delta)$ where $Q$ is a finite set of \emph{states}, $Q_f\subset Q$ is the set of \emph{final states} and $\Delta \subset \bigcup_{n\geq 0}\Sigma_n \times Q^{n+1}$ is a finite set of \emph{transitions}.
The \emph{output} of $\mathcal{A}$, noted $\delta$, is a function from $T(\Sigma)$ to $2^Q$ inductively defined for any tree $t=f(t_1,\ldots,t_n)$ by $\delta(t)=\{q\in Q \mid \exists (f,q_1,\ldots,q_n,q)\in \Delta, (\forall 1\leq i\leq n,q_i\in\delta(t_i))\}$.
The \emph{accepted language} of $\mathcal{A}$ is  $L(\mathcal{A})=\{t\in T(\Sigma)| \delta(t)\cap Q_f\neq\emptyset\}$.
The \emph{state language} $L(q)$ (also known as down language~\cite{conf/stringology/CleophasKSW09}) of a state $q\in Q$ is defined by $L(q)=\{t \in  T(\Sigma) |  q\in \delta(t)\}$.
Obviously, 
\begin{align}
  L(\mathcal{A})=\bigcup_{q\in Q_f}L(q) \label{eq lang et lang bas}
\end{align}
In the following of this paper, we consider \emph{accessible} FTAs, that are FTAs any state $q$ of which satisfies $L(q)\neq\emptyset$.
Obviously, any FTA admits an equivalent accessible FTA obtained by removing the states the down language of which is empty.

Given a symbol $c$ in $\Sigma_0$, the $c$-\emph{product} is the operation $\cdot_c$ defined for any tree $t$ in $T(\Sigma)$ and for any tree language $L$ by
\begin{equation}
t\cdot_c L=
\left\{
  \begin{array}{l@{\ }l}
    L & \text{ if }t=c,\\
    \{d\} & \text{ if }t=d\in\Sigma_0\setminus\{c\},\\
    f(t_1\cdot_c L,\ldots,t_n\cdot_c L ) & \text{ otherwise if } t=f(t_1,\ldots,t_n)\\
  \end{array}
\right.\label{eq def lang sub}
\end{equation}
This $c$-product is extended for any two tree languages $L$ and $L'$ by $L\cdot_c L'=\bigcup_{t\in L} t\cdot_c L'$.
In the following of this paper, we use some equivalences over expressions using some properties of the $c$-product.
Let us state these properties of the $c$-product.
As it is the case of catenation product in the string case, it distributes over the union:

\begin{lemma}\label{lem distrib prod sum}
  Let $L_1$, $L_2$ and $L_3$ be three tree languages over $\Sigma$.
  Let $c$ be a symbol in $\Sigma_0$.
  Then:
  \begin{align*}
    (L_1\cup L_2)\cdot_c L_3 &= (L_1 \cdot_c L_3) \cup (L_2\cdot_c L_3)
  \end{align*}
\end{lemma}
\begin{proof}
  Let $t$ be a tree in $T(\Sigma)$.
  Then:
  \begin{align*}
    t\in (L_1\cup L_2)\cdot_c L_3 & \Leftrightarrow \exists u\in L_1\cup L_2, \exists v\in L_3, t=u\cdot_c v \\
    & \Leftrightarrow (\exists u\in L_1, \exists v\in L_3, t=u\cdot_c v )\vee (\exists u\in L_2, \exists v\in L_3, t=u\cdot_c v) \\
    & \Leftrightarrow t\in (L_1 \cdot_c L_3) \cup (L_2\cdot_c L_3)
  \end{align*}
  \qed
\end{proof}

Another common property with the catenation product is that any operator $\cdot_c$ is associative:
\begin{lemma}
  Let $t$ and $t'$  be any two trees in $T(\Sigma$), let $L$ be a tree language over $\Sigma$ and let $c$ be a symbol in $\Sigma_0$.
  Then:
  \begin{align*} 
    t\cdot_c(t'\cdot_c L)& = (t\cdot_c t')\cdot_c L
  \end{align*}
\end{lemma}
\begin{proof}
  By induction over the structure of $t$.
  \begin{enumerate}
    \item Consider that $t=c$. 
    Then $t\cdot_c(t'\cdot_c L)=t'\cdot_c L=(t\cdot_c t')\cdot_c L$.
    \item Consider that $t\in\Sigma_0\setminus\{c\}$. 
    Then $t\cdot_c(t'\cdot_c L)=t=(t\cdot_c t')\cdot_c L$.
    \item Let us suppose that $t=f(t_1,\ldots,t_n)$ with $n>0$. Then, following Equation~\eqref{eq def lang sub}:
    \begin{align*}
      f(t_1,\ldots,t_n) \cdot_c(t'\cdot_c L)& = f(t_1\cdot_c(t'\cdot_c L),\ldots,t_n\cdot_c(t'\cdot_c L)) \\
        & = f((t_1\cdot_c t')\cdot_c L,\ldots,(t_n\cdot_c t')\cdot_c L) & \text{(Induction hypothesis)}\\
        & =f(t_1\cdot_c t',\ldots,t_n\cdot_c t')\cdot_c L \\
        & =(f(t_1,\ldots,t_n)\cdot_c t')\cdot_c L 
    \end{align*}
  \end{enumerate}
  \qed
\end{proof}
\begin{corollary}\label{cor cdotc assoc}
  Let $L$, $L'$ and $L''$ be any three tree languages over a graded alphabet $\Sigma$ and let $c$ be a symbol in $\Sigma_0$.
  Then:
  \begin{align*}
    L\cdot_c(L'\cdot_c L'') &= (L\cdot_c L')\cdot_c L''
  \end{align*}
\end{corollary}

However, the associativity is not necessarily satisfied if the substitution symbols are different; as an example, $(f(a,b)\cdot_a b)\cdot_b c\neq f(a,b) \cdot_a (b\cdot_b c)$.
Finally, the final common property is that the operation $\cdot_c$ is compatible with the inclusion:
\begin{lemma}
  Let $t$ be a tree over $\Sigma$, and let $L\subset L'$ be two tree languages over $\Sigma$.
  Then:
  \begin{align*}
    t\cdot_c L & \subset t\cdot_c L'
  \end{align*}
\end{lemma}
\begin{proof}  
  By induction over the structure of $t$.
  \begin{enumerate}
    \item Consider that $t=c$. 
    Then  $c\cdot_c L =L \subset L'=c\cdot_c L'$.
    \item Consider that $t\in\Sigma_0\setminus\{c\}$. 
    Then $t \cdot_c L =\{t\}=t\cdot_c L'$.
    \item Let us suppose that $t=f(t_1,\ldots,t_n)$.
    \begin{align*}
      \intertext{Then}
      f(t_1,\ldots,t_n) \cdot_c L & = f(t_1\cdot_c L ,\ldots,t_n\cdot_c L)
      \intertext{By induction hypothesis,}
      \forall 1\leq j\leq n, t_j\cdot_c L & \subset t_j \cdot_c L' 
      \intertext{Therefore,}
      f(t_1\cdot_c L ,\ldots,t_n\cdot_c L)& \subset f(t_1\cdot_c L' ,\ldots,t_n\cdot_c L')=t\cdot_c L'
    \end{align*}
  \end{enumerate}
  \qed
\end{proof}
\begin{corollary}\label{cor cdotc comp incl}
  Let $L$, $L'\subset L''$ be any three tree languages over $\Sigma$ and let $c$ be a symbol in $\Sigma_0$.
  Then:
  \begin{align*}
    L\cdot_c L' & \subset L\cdot_c L''
  \end{align*}
\end{corollary}

The first property not shared with the classical catenation product is that the $c$-product may distribute over other products:
\begin{lemma}\label{lem distribut sub}
  Let $t_1$, $t_2$ and $t_3$ be any three trees in $T(\Sigma)$.
  Let $a$ and $b$ be two distinct symbols in $\Sigma_0$ such that $a$ does not appear in $t_3$.
  Then:
  \begin{align*}
    (t_1\cdot_a t_2)\cdot_b t_3 & =(t_1 \cdot_b t_3)\cdot_a (t_2\cdot_b t_3)
  \end{align*}
\end{lemma}
\begin{proof}
  By induction over $t_1$.
  \begin{enumerate}
    \item If $t_1=a$, then 
    \begin{align*}
      (t_1\cdot_a t_2)\cdot_b t_3 = t_2 \cdot_b t_3 =(t_1 \cdot_b t_3)\cdot_a (t_2\cdot_b t_3)
    \end{align*}
    \item If $t_1=b$, then
    \begin{align*}
      (t_1\cdot_a t_2)\cdot_b t_3  = t_3 =(t_1 \cdot_b t_3)\cdot_a (t_2\cdot_b t_3) 
    \end{align*}
    \item If $t_1=c\in\Sigma_0\setminus\{a,b\}$, then
    \begin{align*}
      (t_1\cdot_a t_2)\cdot_b t_3  = t_1 =(t_1 \cdot_b t_3)\cdot_a (t_2\cdot_b t_3) 
    \end{align*}
    \item If $t_1=f(u_1,\ldots,u_n)$ with $n>0$, then, following Equation~\eqref{eq def lang sub}:
    \begin{align*}
      (t_1\cdot_a t_2)\cdot_b t_3 & = (f(u_1\cdot_a t_2,\ldots,u_n\cdot_a t_2))\cdot_b t_3\\
      & = f((u_1\cdot_a t_2)\cdot_b t_3,\ldots,(u_n\cdot_a t_2)\cdot_b t_3) \\
      & = f((u_1\cdot_b t_3)\cdot_a (t_2\cdot_b t_3),\ldots,(u_n\cdot_b t_3)\cdot_a (t_2\cdot_b t_3)) & \text{(Induction Hypothesis)}\\
      & = f(u_1\cdot_b t_3,\ldots,u_n\cdot_b t_3)\cdot_a (t_2\cdot_b t_3) \\
      & = (f(u_1,\ldots,u_n)\cdot_b t_3)\cdot_a (t_2\cdot_b t_3)
    \end{align*}
  \end{enumerate}
  \qed
\end{proof}

\begin{corollary}\label{cor distribut sub}
  Let $L_1$, $L_2$ and $L_3$ be any three tree languages over $\Sigma$.
  Let $a$ and $b$ be two distinct symbols in $\Sigma_0$ such that $L_3\subset T(\Sigma\setminus\{a\})$.
  Then:
  \begin{align*}
    (L_1\cdot_a L_2)\cdot_b L_3 & =(L_1 \cdot_b L_3)\cdot_a (L_2\cdot_b L_3)
  \end{align*}
\end{corollary}

In some particular cases, two products commute:
\begin{lemma}\label{lem commut sub}
  Let $t_1$, $t_2$ and $t_3$ be any three trees in $T(\Sigma)$.
  Let $a$ and $b$ be two distinct symbols in $\Sigma_0$ such that $a$ does not appear in $t_3$ and such that $b$ does not appear in $t_2$.
  Then:
  \begin{align*}
    (t_1\cdot_a t_2)\cdot_b t_3 & =(t_1 \cdot_b t_3)\cdot_a t_2
  \end{align*}
\end{lemma}
\begin{proof}
  By induction over $t_1$.
  \begin{enumerate}
    \item If $t_1=a$, then 
    \begin{alignat*}{2}
      (t_1\cdot_a t_2)\cdot_b t_3 & = (a\cdot_a t_2) \cdot_b t_3 & & = t_2 \cdot_b t_3\\
       & = t_2 & & = a \cdot_a t_2\\
       & = (a \cdot_b t_3) \cdot_a t_2 & & =(t_1 \cdot_b t_3)\cdot_a t_2 
    \end{alignat*}
    \item If $t_1=b$, then
    \begin{alignat*}{2}
      (t_1\cdot_a t_2)\cdot_b t_3 & = (b\cdot_a t_2)\cdot_b t_3&& = b \cdot_b t_3\\
      & = t_3 && =  t_3\cdot_a t_2 \\
      & = (b \cdot_b t_3)\cdot_a t_2 && = (t_1 \cdot_b t_3)\cdot_a t_2)
    \end{alignat*}
    \item If $t_1=c\in\Sigma_0\setminus\{a,b\}$, then
    \begin{alignat*}{2}
      (t_1\cdot_a t_2)\cdot_b t_3 & = (c \cdot_a t_2)\cdot_b t_3 && = c \cdot_b t_3\\
      & = c && = c \cdot_a t_2\\
      & = (c \cdot_b t_3) \cdot_a t_2
    \end{alignat*}
    \item If $t_1=f(u_1,\ldots,u_n)$ then, following Equation~\eqref{eq def lang sub}:
    \begin{align*}
      (t_1\cdot_a t_2)\cdot_b t_3 & = (f(u_1\cdot_a t_2,\ldots,u_n\cdot_a t_2))\cdot_b t_3\\
      & = f((u_1\cdot_a t_2)\cdot_b t_3,\ldots,(u_n\cdot_a t_2)\cdot_b t_3)\\
      & = f((u_1\cdot_b t_3)\cdot_a t_2,\ldots,(u_n\cdot_b t_3)\cdot_a t_2)& \text{(Induction Hypothesis)}\\
      & = f(u_1\cdot_b t_3,\ldots,u_n\cdot_b t_3)\cdot_a t_2\\
      & = (f(u_1,\ldots,u_n)\cdot_b t_3)\cdot_a t_2
    \end{align*}
  \end{enumerate}
  \qed
\end{proof}

The \emph{iterated} $c$-\emph{product} is the operation $^{n,c}$ recursively defined for any integer $n$ by:
\begin{align}
  L^{0,c} & =  \{c\}\label{eq def iter prod 1}\\
  L^{n+1,c} & =  L^{n,c} \cup L\cdot_c  L^{n,c}\label{eq def iter prod 2}
\end{align}
The $c$-\emph{closure} is the operation $^{*_{c}}$ defined by $L^{*_c} = \bigcup_{n\geq 0} L^{n,c}$.
Notice that, unlike the string case, the products may commute with the closure in some cases:
\begin{lemma}\label{lem permut sub star}
  Let $L_1$ and $L_2$ be any two tree languages over $\Sigma$.
  Let $a$ and $b$ be two distinct symbols in $\Sigma_0$ such that $L_2\subset T(\Sigma\setminus\{a\})$.
  Then:
  \begin{align*}
    L_1^{*_a}\cdot_b L_2 & =(L_1\cdot_b L_2)^{*_a}
  \end{align*}
\end{lemma}
\begin{proof}
  Let us show by recurrence over the integer $n$ that $L_1^{n,a}\cdot_b L_2  =(L_1\cdot_b L_2)^{n,a}$.
  \begin{enumerate}
    \item If $n=0$, then, according to Equation~\eqref{eq def iter prod 1}):
		  \begin{align*}
		    L_1^{0,a}\cdot_b L_2  = \{a\} = (L_1\cdot_b L_2)^{0,a} 
		  \end{align*}
    \item If $n>0$, then, following Equation~\eqref{eq def iter prod 2}):
		  \begin{align*}
		    L_1^{n+1,a}\cdot_b L_2  & = (L_1^{n,a}\cdot_a L_1 \cup L_1^{n,a})\cdot_b L_2  \\
		    & = (L_1^{n,a}\cdot_a L_1)\cdot_b L_2 \cup (L_1^{n,a})\cdot_b L_2 &\text{(Lemma~\ref{lem distrib prod sum})}\\
		    & = ((L_1^{n,a}\cdot_b L_2)\cdot_a (L_1\cdot_b L_2)) \cup (L_1^{n,a})\cdot_b L_2 & \text{(Corollary~\ref{cor distribut sub})}\\ 
		    & = ((L_1\cdot_b L_2)^{n,a} \cdot_a (L_1\cdot_b L_2)) \cup (L_1\cdot_b L_2)^{n,a} & \text{(Induction Hypothesis)}\\ 
		    & = (L_1\cdot_b L_2)^{n+1,a} 
		  \end{align*}
  \end{enumerate}
  As a direct consequence, $L_1^{*_a}\cdot_b L_2  =(L_1\cdot_b L_2)^{*_a}$.
  \qed
\end{proof}

A \emph{rational expression} $E$ over $\Sigma$ is inductively defined by:
\begin{align*}
  \begin{gathered}
    \begin{aligned}
      E &=0, & E&=f(E_1,\ldots,E_n),
    \end{aligned}\\
    \begin{aligned}
      E&=E_1+E_2, & E&=E_1\cdot_c E_2,& E&=E_1^{*_c} 
    \end{aligned}
  \end{gathered}
\end{align*}
where $f$ is any symbol in $\Sigma_n$, $c$ is any symbol in $\Sigma_0$ and $E_1,\ldots,E_n$ are any $n$ rational expressions.
The \emph{language denoted by} $E$ is the tree language $L(E)$ inductively defined by:
\begin{align*}
  \begin{gathered}
    \begin{aligned}
      L(0)&=\emptyset, & L(f(E_1,\ldots,E_n))&=f(L(E_1),\ldots,L(E_n)),
    \end{aligned}\\
    \begin{aligned}
      L(E_1+E_2)&=L(E_1)\cup L(E_2), & L(E_1\cdot_c E_2)&=L(E_1)\cdot_c L(E_2),& L(E_1^{*_c})&=(L(E_1))^{*_c} 
    \end{aligned}
  \end{gathered}
\end{align*}
where $f$ is any symbol in $\Sigma_n$, $c$ is any symbol in $\Sigma_0$ and $E_1,\ldots,E_n$ are any $n$ rational expressions.
In the following of this paper, we consider that rational expressions include some \emph{variables}.
Let $X=\{x_1,\ldots,x_k\}$ be a set of $k$ variables.
A rational expression $E$ over $(\Sigma,X)$ is inductively defined by:
\begin{align*}
  \begin{gathered}
    \begin{aligned}
      E &=0, & E&=x_j,& E&=f(E_1,\ldots,E_n),
    \end{aligned}\\
    \begin{aligned}
      E&=E_1+E_2, & E&=E_1\cdot_c E_2,& E&=E_1^{*_c} 
    \end{aligned}
  \end{gathered}
\end{align*}
where $f$ is any symbol in $\Sigma_n$, $c$ is any symbol in $\Sigma_0$, $1\leq j\leq k$ is any integer and $E_1,\ldots,E_n$ are any $n$ rational expressions over $(\Sigma,X)$.
The language denoted by an expression with variables needs a context to be computed: indeed, any variable has to be evaluated according to a tree language.
Let $\mathcal{L}=(L_1,\ldots,L_k)$ be a $k$-tuple of tree languages over $\Sigma$.
The $\mathcal{L}$-language denoted by $E$ is the tree language $L_\mathcal{L}(E)$ inductively defined by:
\begin{align*}
  \begin{gathered}
    \begin{aligned}
      L_\mathcal{L}(0)&=\emptyset, & L_\mathcal{L}(x_j)&=L_j,
    \end{aligned}\\
    \begin{aligned}
      L_\mathcal{L}(f(E_1,\ldots,E_n))&=f(L_\mathcal{L}(E_1),\ldots,L_\mathcal{L}(E_n)),
    \end{aligned}\\
    \begin{aligned}
      L_\mathcal{L}(E_1+E_2)&=L_\mathcal{L}(E_1)\cup L_\mathcal{L}(E_2) 
    \end{aligned}\\
    \begin{aligned}
      L_\mathcal{L}(E_1\cdot_c E_2)&=L_\mathcal{L}(E_1)\cdot_c L_\mathcal{L}(E_2),& L_\mathcal{L}(E_1^{*_c})&=(L_\mathcal{L}(E_1))^{*_c} 
    \end{aligned}
  \end{gathered}
\end{align*}
where $f$ is any symbol in $\Sigma_n$, $c$ is any symbol in $\Sigma_0$, $1\leq j\leq k$ is any integer and $E_1,\ldots,E_n$ are any $n$ rational expressions over $(\Sigma,X)$.
Two rational expressions $E$ and $F$ with variables are \emph{equivalent}, denoted by $E\sim F$, if for any tuple $\mathcal{L}$ of languages over $\Sigma$, $L_\mathcal{L}(E)= L_\mathcal{L}(F)$.
Let $\Gamma\subset\Sigma$. 
Two rational expressions $E$ and $F$ with variables are $\Gamma$-\emph{equivalent}, denoted by $E\sim_\Gamma F$, if for any tuple $\mathcal{L}$ of languages over $\Gamma$, $L_\mathcal{L}(E)= L_\mathcal{L}(F)$.
By definition,
\begin{align}
  E\sim F &\Rightarrow E\sim_\Gamma F\label{eq impl sim alpha}
\end{align}
Notice that any expression over $(\Sigma,X)$ is also an expression over $\Sigma\cup X$.
However, two equivalent rational expressions over $(\Sigma,X)$ are not necessarily equivalent as rational expressions over $\Sigma\cup X$.
As an example, $x\cdot_a b$ is equivalent to $x$ as expressions over $\{a,b,x\}$, but not as expressions over $(\{a,b\},\{x\})$:
\begin{alignat*}{4}
  L(x\cdot_a b) &=\{x\} &&= L(x)\\
  L_{\{a\}}(x\cdot_a b)& =\{b\} &&\neq  L_{\{a\}}(x) &&=\{a\}
\end{alignat*}

In the following, we denote by $E_{x\leftarrow E'}$ the expression obtained by substituting any symbol $x$ by the expression $E'$ in the expression $E$.
Obviously, this transformation is inductively defined as follows:
\begin{align*}
  \begin{gathered}
	  \begin{aligned}
	    a_{x\leftarrow E'} & =a & 0_{x\leftarrow E'} & =0\\
	    y_{x\leftarrow E'} & =y & x_{x\leftarrow E'} & =E'
	  \end{aligned}\\
	  \begin{aligned}
	    (f(E_1,\ldots,E_n))_{x\leftarrow E'}  =f((E_1)_{x\leftarrow E'},\ldots,(E_n)_{x\leftarrow E'})
	  \end{aligned}\\
	  \begin{aligned}
	    (E_1+E_2)_{x\leftarrow E'} & =(E_1)_{x\leftarrow E'}+(E_2)_{x\leftarrow E'} & (E_1\cdot_c E_2)_{x\leftarrow E'} & =(E_1)_{x\leftarrow E'}\cdot_c(E_2)_{x\leftarrow E'}
	  \end{aligned}\\
	  \begin{aligned}
	    (E_1^{*_c})_{x\leftarrow E'}  &=((E_1)_{x\leftarrow E'})^{*_c} 
	  \end{aligned}
  \end{gathered}
\end{align*}
where $a$ is any symbol in $\Sigma_0$, $x\neq y$ are two variables in $X$, $f$ is any symbol in $\Sigma_n$, $c$ is any symbol in $\Sigma_0$ and $E_1,\ldots,E_n$ are any $n$ rational expressions over $(\Sigma,X)$.
This transformation preserves the language in the following case:
\begin{lemma}\label{lem sub cons lang}
  Let $E$ be an expression over an alphabet $\Sigma$ and over a set $X=\{x_1,\ldots,x_n\}$ of variables.
  Let $F$ be a rational expression over $(\Sigma,X)$.
  Let $x_j$ be a variable in $X$.
  Let $\mathcal{L}=(L_1,\ldots,L_n)$ be a $n$-uple of tree languages such that $L_j=L_\mathcal{L}(F)$.
  Then:
  \begin{align*}
    L_\mathcal{L}((E)_{x_j\leftarrow F}) & = L_\mathcal{L}(E)
  \end{align*}
\end{lemma}
\begin{proof}
  By induction over the structure of $E$.
  \begin{enumerate}
    \item If $E\in\{a,y,0\}$ with $a\in\Sigma_0$ and $y\neq x_j$, $(E)_{x_j\leftarrow F}=E$.
    \item If $E=x_j$, then $(E)_{x_j\leftarrow F}=F$.
    Therefore 
    \begin{alignat*}{2}
      L_\mathcal{L}((E)_{x_j\leftarrow F}) &=L_{\mathcal{L}}(F) &&=L_j\\
      &=L_{\mathcal{L}}(x_j) && =L_{\mathcal{L}}(E)
    \end{alignat*}
    \item If $E=f(E_1,\ldots,E_n)$, with $f\in\Sigma_k$, $k>0$ then:
    \begin{align*}
      L_\mathcal{L}((E)_{x_j\leftarrow F}) &=L_{\mathcal{L}}(f((E_1)_{x_j\leftarrow F},\ldots,(E_n)_{x_j\leftarrow F}))\\
      &=f(L_{\mathcal{L}}((E_1)_{x_j\leftarrow F}),\ldots,L_{\mathcal{L}}((E_n)_{x_j\leftarrow F}))\\   
      &=f(L_{\mathcal{L}}(E_1),\ldots,L_{\mathcal{L}}(E_n))  & \text{(Induction Hypothesis)}\\      
      &=L_{\mathcal{L}}(f(E_1,\ldots,E_n))  
    \end{align*}
    \item If $E=E_1+E_2$, then 
    \begin{align*}
      L_\mathcal{L}((E_1+E_2)_{x_j\leftarrow F}) &=L_\mathcal{L}((E_1)_{x_j\leftarrow F}+(E_2)_{x_j\leftarrow F})\\
      &=L_\mathcal{L}((E_1)_{x_j\leftarrow F})\cup L_\mathcal{L}((E_2)_{x_j\leftarrow F}))\\
      &=L_\mathcal{L}(E_1) \cup L_\mathcal{L}(E_2) & \text{(Induction Hypothesis)}\\
      &=L_\mathcal{L}(E_1+E_2)
    \end{align*}
    \item If $E=E_1\cdot_c E_2$, then
    \begin{align*}
      L_\mathcal{L}((E_1\cdot_c E_2)_{x_j\leftarrow F}) &=L_\mathcal{L}((E_1)_{x_j\leftarrow F}\cdot_c(E_2)_{x_j\leftarrow F})\\
      &=L_\mathcal{L}((E_1)_{x_j\leftarrow F})\cdot_c L_\mathcal{L}((E_2)_{x_j\leftarrow F}))\\
      &=L_\mathcal{L}(E_1) \cdot_c L_\mathcal{L}(E_2) & \text{(Induction Hypothesis)}\\
      &=L_\mathcal{L}(E_1 \cdot_c  E_2)
    \end{align*}
    \item If $E=E_1^{*_c}$, then
    \begin{align*}
      L_\mathcal{L}((E_1^{*_c})_{x_j\leftarrow F}) &=(L_\mathcal{L}((E_1)_{x_j\leftarrow F}))^{*_c}\\
      &=(L_\mathcal{L}(E_1))^{*_c} & \text{(Induction Hypothesis)}\\
      &=L_\mathcal{L}(E_1^{*_c})
    \end{align*}
  \end{enumerate}
  \qed
\end{proof}

In the following, we denote by $\mathrm{op}(E)$ the set of the operators that appear in a rational expression $E$.
The previous substitution can be used in order to factorize an expression w.r.t. a variable.
However, this operation does not preserve the equivalence; \emph{e.g.}
\begin{align*}
  L_{\{b\}}(x\cdot_b c)  = \{c\} \neq L_{\{b\}}((a\cdot_b c) \cdot_a x)=\{b\}
\end{align*}

Nevertheless, this operation preserves the language if it is based on a restricted alphabet:
\begin{proposition}\label{prop cas sub var par lettre ok}
  Let $E$ be a rational expression over a graded alphabet $\Sigma$ and over a set $X$ of variables.
  Let $x$ be a variable in $X$.
  Let $\Gamma\subset\Sigma$ be the subset defined by $\Gamma=\{b\in\Sigma_0\mid \{\cdot_b,^{*_b}\}\cap\mathrm{op}(E)\neq\emptyset\}$.
  Let $a$ be a symbol not in $\Sigma$.
  Then:
  \begin{align*}
    E\sim_{\Sigma\setminus\Gamma} (E)_{x\leftarrow a}\cdot_a x
  \end{align*}
\end{proposition}
\begin{proof}
  By induction over the structure of $E$.
  \begin{enumerate}
    \item If $E=x$, then since $x\sim_{\Sigma\cup\{a\}} a\cdot_a x$, it holds from Equation~\eqref{eq impl sim alpha} that $E\sim_{\Sigma\setminus\Gamma} (E)_{x\leftarrow a}\cdot_a x$.
    \item If $E\in\{0\} \cup \Sigma\cup X\setminus\{x\}$, since $x$ does not appear in $E$, it holds $E=E_{x\leftarrow a}$.
    \item If $E=f(E_1,\ldots,E_n)$, then
    \begin{align*}
      (f(E_1,\ldots,E_n))_{x\leftarrow a}\cdot_a x &= f((E_1)_{x\leftarrow a},\ldots,(E_n)_{x\leftarrow a})\cdot_a x\\
      &\sim f((E_1)_{x\leftarrow a}\cdot_a x,\ldots,(E_n)_{x\leftarrow a}\cdot_a x) & \text{(Equation~\eqref{eq def lang sub})}\\
      &\sim_{\Sigma\setminus\Gamma} f(E_1,\ldots,E_n) & \text{(Induction hypothesis)}
    \end{align*}
    \item If $E=E_1+E_2$, then
    \begin{align*}
      (E_1+E_2)_{x\leftarrow a}\cdot_a x &= ((E_1)_{x\leftarrow a}+(E_2)_{x\leftarrow a})\cdot_a x\\
      &\sim ((E_1)_{x\leftarrow a})\cdot_a x+((E_2)_{x\leftarrow a})\cdot_a x & \text{(Lemma~\ref{lem distrib prod sum})}\\
      & \sim_{\Sigma\setminus\Gamma} E_1+ E_2 & \text{(Induction hypothesis})
    \end{align*}
    \item If $E=E_1\cdot_c E_2$, then
    \begin{align*}
      (E_1\cdot_c E_2)_{x\leftarrow a}\cdot_a x &= ((E_1)_{x\leftarrow a}\cdot_c (E_2)_{x\leftarrow a})\cdot_a x\\
      &\sim_{\Sigma} (((E_1)_{x\leftarrow a})\cdot_a x)\cdot_c(((E_2)_{x\leftarrow a})\cdot_a x) & \text{(Corollary~\ref{cor distribut sub})}\\
      & \sim_{\Sigma\setminus\Gamma} E_1 \cdot_c E_2 & \text{(Induction hypothesis})
    \end{align*}
    \item If $E=E_1^{*_c}$, then
    \begin{align*}
      (E_1^{*_c})_{x\leftarrow a}\cdot_a x &= ((E_1)_{x\leftarrow a})^{*_c}\cdot_a x\\
      &\sim_{\Sigma} (((E_1)_{x\leftarrow a})\cdot_a x)^{*_c} & \text{(Lemma~\ref{lem permut sub star})}\\
      & \sim_{\Sigma\setminus\Gamma} E_1^{*_c} & \text{(Induction hypothesis})
    \end{align*}
  \end{enumerate}
	\qed
\end{proof}

\section{Equations Systems for Tree Languages}\label{sec eq syst}

Let $\Sigma$ be an alphabet and $\mathbb{E}=\{\mathbb{E}_1,\ldots,\mathbb{E}_n\}$ be a set of $n$ variables.
An \emph{equation} over $(\Sigma,\mathbb{E})$ is an expression $\mathbb{E}_j=F_j$, where $1\leq j\leq n$ is any integer and $F_j$ is a rational expression over $(\Sigma,\mathbb{E})$.
An \emph{equation system} over $(\Sigma,\mathbb{E})$ is a set $\mathcal{X}=\{\mathbb{E}_j=F_j \mid 1\leq j\leq n\}$ of $n$ equations.
Let $\mathcal{L}=(L_1,\ldots,L_n)$ be a $n$-tuple of tree languages.
The tuple $\mathcal{L}$ is a \emph{solution} for an equation $(\mathbb{E}_j=F_j)$ if $L_j=L_{\mathcal{L}}(F_j)$.
The tuple $\mathcal{L}$ is a \emph{solution} for $\mathcal{X}$ if for any equation $(\mathbb{E}_j=F_j)$ in $\mathcal{X}$,  $\mathcal{L}$ is a solution of $(\mathbb{E}_j=F_j)$.
\begin{example}\label{ex syst eq}
  Let us define the equation system $\mathcal{X}$ as follows:
  \begin{align*}
  \mathcal{X}&=
    \begin{cases}
      \mathbb{E}_1 & =f(\mathbb{E}_1,\mathbb{E}_1)+f(\mathbb{E}_2,\mathbb{E}_4)\\
      \mathbb{E}_2 & =b+f(\mathbb{E}_2,\mathbb{E}_4)\\
      \mathbb{E}_3 & =a+h(\mathbb{E}_4)\\
      \mathbb{E}_4 & =a+h(\mathbb{E}_3)
    \end{cases}
  \end{align*}
  The tuple $(\emptyset,\emptyset,\emptyset,\emptyset)$ is a solution for the equation $\mathbb{E}_1=F_1$, but not of the system $\mathcal{X}$.
\end{example}
Two systems over the same variables are \emph{equivalent} if they admit the same solutions.
Notice that a system does not necessarily admit a unique solution.
As an example, any language is a solution of the system $\mathbb{E}_1=\mathbb{E}_1$.
Obviously, 
\begin{proposition}\label{prop pas de var sol un}
  If $\mathcal{X}$ only contains equations $\mathbb{E}_k=F_k$ with $F_k$ a rational expression without variables, then $(L(F_1),\ldots,L(F_n))$ is the unique solution of $\mathcal{X}$.
\end{proposition}

Let us now define the operation of substitution, computing an equivalent system.
\begin{definition}
  Let $\mathcal{X}=\{(\mathbb{E}_j=F_j)\mid 1\leq j\leq n\}$ be an equation system.
  The \emph{substitution of} $(\mathbb{E}_k=F_k)$ in $\mathcal{X}$ is the system $\mathcal{X}^k=\{\mathbb{E}_k=F_k\}\cup\{\mathbb{E}_j=(F_j)_{\mathbb{E}_k\leftarrow F_k} \mid j\neq k\wedge 1\leq j\leq n\}$.
\end{definition}

As a direct consequence of Lemma~\ref{lem sub cons lang},

\begin{proposition}\label{prop sub ok dans syst}
  Let $\mathcal{X}=\{(\mathbb{E}_j=F_j)\mid 1\leq j\leq n\}$ be an equation system.
  Let $\mathbb{E}_k=F_k$ be an equation in $\mathcal{X}$.
  Let $\mathcal{L}$ be a solution of $\mathcal{X}$.
  Then for any integer $1\leq j,k\leq n$ with $j\neq k$,
  \begin{equation*}
    \mathcal{L}\text{ is a solution of }\mathbb{E}_j=(F_j)_{\mathbb{E}_k\leftarrow F_k}.
  \end{equation*} 
\end{proposition}

And following Proposition~\ref{prop sub ok dans syst},
\begin{proposition}\label{prop sub cons sol}
  Let $\mathcal{X}$ be an equation system over $n$ variables.
  Let $k\leq n$ be an integer.
  Then:
  \begin{equation*}
    \mathcal{X}\text{ and } \mathcal{X}^k\text{ are equivalent.}
  \end{equation*} 
\end{proposition}

\begin{example}\label{ex syst substi}
  Let us consider the system $\mathcal{X}$ of Example~\ref{ex syst eq}. Then:
  \begin{align*}
  \mathcal{X}^4&=
    \begin{cases}
      \mathbb{E}_1 & =f(\mathbb{E}_1,\mathbb{E}_1)+f(\mathbb{E}_2,a+h(\mathbb{E}_3))\\
      \mathbb{E}_2 & =b+f(\mathbb{E}_2,a+h(\mathbb{E}_3))\\
      \mathbb{E}_3 & =a+h(a+h(\mathbb{E}_3))\\
      \mathbb{E}_4 & =a+h(\mathbb{E}_3)
    \end{cases}
  \end{align*}
\end{example}

Let us determine a particular case that can be solved by successive substitutions.
Let $\mathcal{X}=\{(\mathbb{E}_j=F_j)\mid 1\leq j\leq n\}$ be an equation system.
The relation $<_\mathcal{X}$ is defined for any two variables $\mathbb{E}_j$ and $\mathbb{E}_k$ by
  \begin{align*}
    \mathbb{E}_j <_\mathcal{X} \mathbb{E}_k   \Leftrightarrow &\  \mathbb{E}_j\text{ appears in }F_k
  \end{align*}
The relation $\preceq_\mathcal{X}$ is defined as the transitive closure of $<_\mathcal{X}$.
In the case where $\mathbb{E}_k<_\mathcal{X} \mathbb{E}_k$ , the equation $\mathbb{E}_k=F_k$ is said to be \emph{recursive}. 
Let us say that a system is \emph{recursive} if there exists two symbols $\mathbb{E}_j$ and $\mathbb{E}_k$ such that $\mathbb{E}_j\preceq_\mathcal{X} \mathbb{E}_k$ and $\mathbb{E}_k\preceq_\mathcal{X} \mathbb{E}_j$.
If a system is not recursive, it can be solved by successive substitutions.
If $\mathbb{E}_k$ is a variable that does not appear in any right side of an equation of $\mathcal{X}$, we denote by $\mathcal{X}\setminus (\mathbb{E}_k=F_k)$ the system obtained by removing $\mathbb{E}_k=F_k$ of $\mathcal{X}$, and by reindexing any symbol $\mathbb{E}_j$ with $j>k$ into $\mathbb{E}_{j-1}$.

\begin{lemma}\label{lem sub cas eq pas rec}
  Let $\mathcal{X}=\{(\mathbb{E}_j=F_j)\mid 1\leq j\leq n\}$ be an equation system over a graded alphabet $\Sigma$ and over $n$ variables $\{\mathbb{E}_1,\ldots,\mathbb{E}_n\}$.
  Let $\mathbb{E}_k=F_k$ be an equation in $\mathcal{X}$ such that $\mathbb{E}_k=F_k$ is not recursive. 
  Then for any $n-1$-tuple $Z=(L_1,\ldots,L_{k-1},L_{k+1},\ldots,L_n)$, the two following conditions are equivalent:
  \begin{enumerate}
    \item $(L_1,\ldots,L_{k-1},L_{Z}(F_k),L_{k+1},\ldots,L_n)$ is a solution of $\mathcal{X}$
    \item $(L_1,\ldots,L_{k-1},L_{k+1},\ldots,L_n)$ is a solution of $\mathcal{X}^k\setminus\{\mathbb{E}_k=F_k\}$
  \end{enumerate}
\end{lemma}
\begin{proof}
  Let $\mathcal{L}=(L_1,\ldots,L_{k-1},L_{Z}(F_k),L_{k+1},\ldots,L_n)$ and $\mathcal{L}'=(L_1,\ldots,L_{k-1},L_{k+1},\ldots,L_n)$.
  Obviously, $\mathcal{L}$ is a solution for the (non recursive) equation $\mathbb{E}_k=F_k$.
  From Proposition~\ref{prop sub cons sol},
  \begin{align*}
    \mathcal{L}\text{ is a solution of }\mathcal{X} & \Leftrightarrow \mathcal{L}\text{ is a solution of }\mathcal{X}^k
    \intertext{Consequently, for any integer $j\neq k$,}
    \mathcal{L}\text{ is a solution of }\mathbb{E}_j=F_j & \Leftrightarrow \mathcal{L}\text{ is a solution of }\mathbb{E}_j=(F_j)_{\mathbb{E}_k\leftarrow F_k}
    \intertext{Moreover, by definition of $\mathcal{L}'$, for any integer $j\neq k$,}
    \mathcal{L}\text{ is a solution of }\mathbb{E}_j=(F_j)_{\mathbb{E}_k\leftarrow F_k} & \Leftrightarrow \mathcal{L}'\text{ is a solution of }\mathbb{E}_j=(F_j)_{\mathbb{E}_k\leftarrow F_k}\\
    & \Leftrightarrow \mathcal{L}'\text{ is a solution of }\mathcal{X}^k\setminus\{\mathbb{E}_k=F_k\}
  \end{align*}
  \qed
\end{proof}

As a direct consequence of the previous lemma, a non-recursive system can be solved by solving a smaller system, obtained by substitution:
\begin{corollary}\label{cor sub sol}
  Let $\mathcal{X}=\{(\mathbb{E}_j=F_j)\mid 1\leq j\leq n\}$ be an equation system over a graded alphabet $\Sigma$ and over $n$ variables $\{\mathbb{E}_1,\ldots,\mathbb{E}_n\}$.
  Let $\mathbb{E}_k=F_k$ be an equation in $\mathcal{X}$ such that $F_k$ is a rational expression. 
  Then for any $n-1$-tuple $(L_1,\ldots,L_{k-1},L_{k+1},\ldots,L_n)$, the two following conditions are equivalent:
  \begin{enumerate}
    \item $(L_1,\ldots,L_{k-1},L(F_k),L_{k+1},\ldots,L_n)$ is a solution of $\mathcal{X}$
    \item $(L_1,\ldots,L_{k-1},L_{k+1},\ldots,L_n)$ is a solution of $\mathcal{X}^k\setminus\{\mathbb{E}_k=F_k\}$
  \end{enumerate}
\end{corollary}

Moreover, such a system admits a unique solution:
\begin{proposition}
  Let $\mathcal{X}=\{(\mathbb{E}_j=F_j)\mid 1\leq j\leq n\}$ be an equation system that is not recursive over a graded alphabet $\Sigma$ and over variables $\{\mathbb{E}_1,\ldots,\mathbb{E}_n\}$.
  Then
  \begin{equation*}
    \mathcal{X}\text{ admits a unique solution}.
  \end{equation*}
\end{proposition}
\begin{proof}
  By recurrence over the cardinal of $\mathcal{X}$.
  \begin{enumerate}
    \item $\mathcal{X}=\{\mathbb{E}_1=F_1\}$, then $F_1$ is a rational expression over $\Sigma$ (with no variable) and therefore $L(F_1)$ is the unique solution of $\mathcal{X}$.
    \item Since $\mathcal{X}$ is not recursive, there exists an equation $\mathbb{E}_k=F_k$ with $F_k$ a rational expression over $\Sigma$ (with no variable).
    Therefore, according to Corollary~\ref{cor sub sol}, a tuple $(L_1,\ldots,L_{k-1},L(F_k),L_{k+1},\ldots,L_n)$ is a solution of $\mathcal{X}$ if and only if $(L_1,\ldots,L_{k-1},L_{k+1},\ldots,L_n)$ is a solution of $\mathcal{X}^k\setminus\{\mathbb{E}_k=F_k\}$. 
    By recurrence hypothesis, since $\mathcal{X}^k\setminus\{\mathbb{E}_k=F_k\}$ is not recursive, it admits a unique solution $(L_1,\ldots,L_{k-1},L_{k+1},\ldots,L_n)$.
    Thus $(L_1,\ldots,L_{k-1},L(F_k),L_{k+1},\ldots,L_n)$ is a solution of $\mathcal{X}$.
    Finally, since for any $L_k\neq L(F_k)$,the tuple $(L_1,\ldots,L_{k-1},L_k,L_{k+1},\ldots,L_n)$ is not a solution for $\mathbb{E}_k=F_k$, $(L_1,\ldots,L_{k-1},L(F_k),L_{k+1},\ldots,L_n)$ is the unique solution of $\mathcal{X}$.
  \end{enumerate}
  \qed
\end{proof}

\begin{example}
  Let us define the equation system $\mathcal{Y}$ as follows:
  \begin{align*}
  \mathcal{Y}&=
    \begin{cases}
      \mathbb{E}_1 & =f(\mathbb{E}_2,\mathbb{E}_3)+f(\mathbb{E}_2,\mathbb{E}_3)\\
      \mathbb{E}_2 & =b+f(\mathbb{E}_4,\mathbb{E}_4)\\
      \mathbb{E}_3 & =a+h(\mathbb{E}_4)\\
      \mathbb{E}_4 & =a+(f(a,b))^{*_b}\cdot_b a
    \end{cases}
  \end{align*}
  Then
  \begin{align*}
	  \mathcal{Y}^4&=
	    \begin{cases}
	      \mathbb{E}_1 & =f(\mathbb{E}_2,\mathbb{E}_3)+f(\mathbb{E}_2,\mathbb{E}_3)\\
	      \mathbb{E}_2 & =b+f(a+(f(a,b))^{*_b}\cdot_b a,a+(f(a,b))^{*_b}\cdot_b a)\\
	      \mathbb{E}_3 & =a+h(a+(f(a,b))^{*_b}\cdot_b a)\\
	      \mathbb{E}_4 & =a+(f(a,b))^{*_b}\cdot_b a
	    \end{cases}\\
	  (\mathcal{Y}^4)^3&=
	    \begin{cases}
	      \mathbb{E}_1 & =f(\mathbb{E}_2,a+h(a+(f(a,b))^{*_b}\cdot_b a))+f(\mathbb{E}_2,a+h(a+(f(a,b))^{*_b}\cdot_b a))\\
	      \mathbb{E}_2 & =b+f(a+(f(a,b))^{*_b}\cdot_b a,a+(f(a,b))^{*_b}\cdot_b a)\\
	      \mathbb{E}_3 & =a+h(a+(f(a,b))^{*_b}\cdot_b a)\\
	      \mathbb{E}_4 & =a+(f(a,b))^{*_b}\cdot_b a
	    \end{cases}\\
	  ((\mathcal{Y}^4)^3)^2&=
	    \begin{cases}
	      \mathbb{E}_1 & =f(b+f(a+(f(a,b))^{*_b}\cdot_b a,a+(f(a,b))^{*_b}\cdot_b a),a+h(a+(f(a,b))^{*_b}\cdot_b a))\\
	      & \quad +f(b+f(a+(f(a,b))^{*_b}\cdot_b a,a+(f(a,b))^{*_b}\cdot_b a),a+h(a+(f(a,b))^{*_b}\cdot_b a))\\
	      \mathbb{E}_2 & =b+f(a+(f(a,b))^{*_b}\cdot_b a,a+(f(a,b))^{*_b}\cdot_b a)\\
	      \mathbb{E}_3 & =a+h(a+(f(a,b))^{*_b}\cdot_b a)\\
	      \mathbb{E}_4 & =a+(f(a,b))^{*_b}\cdot_b a
	    \end{cases}
  \end{align*}
\end{example}

\section{Arden's Lemma for Trees and Recursive Systems}\label{ar}

Arden's Lemma~\cite{Ard61} is a fundamental result in automaton theory. 
It gives a solution of the recursive language equation $X=A\cdot X \cup B$ where $X$ is an unknown language.
It can be applied to compute a rational expression from an automaton and therefore prove the second way of Kleene theorem for strings.
Following the same steps as in string case, we generalize this lemma to trees.

\begin{proposition}\label{prop arden}
  Let $A$ and $B$ be two tree languages over a graded alphabet $\Sigma$.
  Then $A^{*_c}\cdot_c B$ is the smallest language in the family $\mathcal{F}$ of languages $L$ over $\Sigma$ satisfying $L=A\cdot_c L \cup B$.
  Furthermore, if $c \notin A$, then $\mathcal{F}=\{A^{*_c}\cdot_c B\}$.
\end{proposition}
\begin{proof}
  Let us set $Z=A^{*_c}\cdot_c B$.
  \begin{enumerate}
	  \item Obviously, $Z$ belongs to $\mathcal{F}$:
	  \begin{align*}
	    A \cdot_c (A^{*_c}\cdot_c B)\cup B 
	      & = (A \cdot_c A^{*_c})\cdot_c B\cup B & \text{from Corollary}~\ref{cor cdotc assoc}\\
	      & = (A \cdot_c A^{*_c})\cdot_c B \cup \{c\} \cdot_c B\\
	      & = ((A \cdot_c A^{*_c}) \cup \{c\} ) \cdot_c B\\
	      & = A^{*_c}\cdot_c B
	  \end{align*}
	  
	  \item Let us now show that if $C$ belongs to $\mathcal{F}$, then $Z\subset C$.
	  To do so, let us show that for any integer $n\geq 0$, $A^{n,c}\cdot_c B\subset C$.
	  Since $C$ belongs to $\mathcal{F}$, then $C=A\cdot_c C \cup B$.
	  Therefore $A^{0,c}\cdot_c B=B\subset C$ and $A\cdot_c C \subset C$. 
	  Suppose that $A^{n,c}\cdot_c B  \subset C$ for some integer $n\geq 0$.
	  Therefore, from Corollary~\ref{cor cdotc comp incl}, $A\cdot_c (A^{n,c}\cdot_c B)  \subset A\cdot_c (C)$ and from Corollary~\ref{cor cdotc assoc}, $A^{n+1,c}\cdot_c B  \subset A\cdot_c C \subset C$.
	  Consequently, since for any integer $n$, $A^{n,c}\cdot_c B\subset C$, it holds that $Z=A^{*_c}\cdot_c B \subset C$.	  
	  \item Finally, let us show that if $c\notin A$, then any language $Y$ in $\mathcal{F}$ satisfies $Y\subset Z$, implying that $\mathcal{F}=\{Z\}$.  
	  Let $Y \neq Z $ satisfying $Y=A\cdot_c Y\cup B$. 
	  Suppose that $Y\not\subset Z$.
	  Let $t$ be a tree in $Y\setminus Z$ such that $\mathrm{Height}(y)$ is minimal.
	  Obviously, since $B\subset Z$, $t$ is not in $B$.
	  Consequently, $t$ belongs to $A\cdot_c Y$ and therefore $t=t_1\cdot_c t_2$ with $t_1\in A$ and $t_2\in Y$.
	  Since $c\notin A$, $t_1\neq c$.
	  Furthermore, if $c$ does not appear in $t_1$, then $t=t_1\in A$ and consequently, $t\in A^{*_c }\cdot_c B=Z$, contradicting the fact that $t\notin Z$.
	  Therefore $c$ appears in $t_1$ and then $\mathrm{Height}(t_2)<\mathrm{Height}(t)$, contradicting the minimality of the height of $t$.
	  As a direct consequence, any language $Y$ in $\mathcal{F}$ satisfies $Y\subset Z$.
	  Following previous point, since $Z\subset Y$, it holds that $Y=Z$.
	\end{enumerate}
  \qed
\end{proof}

By successive substitutions, any recursive system can be transformed into another equivalent system such that there exists a symbol $\mathbb{E}_j$ satisfying $\mathbb{E}_j<_\mathcal{X} \mathbb{E}_j$.
Let us enlighten a specific case where recursive equations can be solved.

For an integer $k$, the $k$-\emph{split} of an expression $F$ over $(\Sigma,\{\mathbb{E}_1,\ldots,\mathbb{E}_n\})$ is the couple $k\mathrm{-split}(F)$ inductively defined by:
\begin{align*}
  k\mathrm{-split}(F) & =
  \begin{cases}
    (E_1'+E'_2,E''_1+E''_2) &\text{ if }F=E_1+E_2 \\
    & \quad \wedge k\mathrm{-split}(E_1)=(E'_1,E''_1) \wedge k\mathrm{-split}(E_2)=(E'_2,E''_2),\\
    (F,0) & \text{ otherwise if }\mathbb{E}_k\text{ appears in }F,\\
    (0,F) & \text{ otherwise.}
  \end{cases}
\end{align*}
Obviously, if $k\mathrm{-split}(F)=(F',F'')$, $F\sim F'+F''$.
This tuple can be used to factorize a recursive equation in order to apply Arden's Lemma.
Indeed, as a direct consequence of Proposition~\ref{prop cas sub var par lettre ok},
\begin{proposition}\label{prop sub par split ok}
  Let $\mathcal{X}=\{(\mathbb{E}_j=F_j)\mid 1\leq j\leq n\}$ be an equation system.
  Let $a$ be a symbol not in $\Sigma$.
  Let $1\leq k\leq n$ be an integer.
  Let $\Gamma\subset\Sigma$ be the subset defined by $\Gamma=\{c\in\Sigma_0\mid \{\cdot_c,^{*_c}\}\cap\mathrm{op}(F_k)\neq\emptyset\}$.
  Let $\mathcal{L}$ be a $n$-tuple of tree languages over the alphabet $\Sigma\setminus\Gamma$.
  Let $k\mathrm{-split}(F)=(F'_k,F''_k)$.
  Then the two following conditions are equivalent:
  \begin{enumerate}
    \item $\mathcal{L}$ is a solution for $\mathbb{E}_k=F_k$,
    \item $\mathcal{L}$ is a solution for $\mathbb{E}_k=(F'_k)_{\mathbb{E}_k\leftarrow a}\cdot_a \mathbb{E}_k+F''_k$.
  \end{enumerate}
\end{proposition}

Once an equation factorized, the Arden's Lemma can be applied by \emph{contraction}:
\begin{definition}
  Let $\mathcal{X}=\{(\mathbb{E}_j=F_j)\mid 1\leq j\leq n\}$ be an equation system.
  Let $1\leq k\leq n$ be an integer such that $\mathbb{E}_k=F'_k\cdot_c \mathbb{E}_k+F''_k$.
  The \emph{contraction of} $(\mathbb{E}_k=F_k)$ in $\mathcal{X}$ is the system $\mathcal{X}_{k}=\{\mathbb{E}_k=(F'_k)^{*_c}\cdot_c F''_k)\}\cup\{\mathbb{E}_j=F_j\mid j\neq k\wedge 1\leq j\leq n\}$.
\end{definition}

Following Proposition~\ref{prop arden}, such a contraction preserves the language:
\begin{proposition}\label{prop sol pour eq rec}
  Let $\mathcal{X}=\{(\mathbb{E}_j=F_j)\mid 1\leq j\leq n\}$ be an equation system.
  Let $1\leq k\leq n$ be an integer such that $\mathbb{E}_k=F'_k\cdot_c \mathbb{E}_k+F''_k$.
  Let $\mathcal{L}=(L_1,\ldots,L_n)$ be a $n$-tuple of tree languages.
  Then the two following conditions are equivalent:
  \begin{enumerate}
    \item $\mathcal{L}$ is a solution of $\mathcal{X}$,
    \item $\mathcal{L}$ is a solution of $\mathcal{X}_k$.
  \end{enumerate} 
  Furthermore, if $c$ is not in $L_\mathcal{L}(F'_k)$ then for any language $L'_k\neq L_k$,
  \begin{equation*}
    (L_1,\ldots,L_{k-1},L'_k,L_{k+1},\ldots,L_n)\text{ is not a solution of }\mathcal{X}.
  \end{equation*} 
\end{proposition}

\begin{example}
  Let us consider the system $\mathcal{X}_4$ of Example~\ref{ex syst substi}:
  \begin{align*}
  \mathcal{X}^4&=
    \begin{cases}
      \mathbb{E}_1 & =f(\mathbb{E}_1,\mathbb{E}_1)+f(\mathbb{E}_2,a+h(\mathbb{E}_3))\\
      \mathbb{E}_2 & =b+f(\mathbb{E}_2,a+h(\mathbb{E}_3))\\
      \mathbb{E}_3 & =a+h(a+h(\mathbb{E}_3))\\
      \mathbb{E}_4 & =a+h(\mathbb{E}_3)
    \end{cases}
  \end{align*}
  The $2-\mathrm{split}$ of $b+f(\mathbb{E}_2,a+h(\mathbb{E}_3))$ is $f(x_2,a+h(\mathbb{E}_3))\cdot_{x_2}\mathbb{E}_2 +b$, contracted in $f(x_2,a+h(\mathbb{E}_3))^{*_{x_2}}\cdot_{x_2} b$.
\end{example}

However, as it was recalled in Proposition~\ref{prop sub par split ok}, the factorization that precedes a contraction does not necessarily produce an equivalent expression.
Let us now define a sufficient property in order to detect solvable systems.
Obviously, it is related to the symbols that appear in a product or a closure.

The \emph{scope} of an operator is its operands.
An occurrence of a symbol $c$ in $\Sigma_0$ is said to be bounded if it appears in the scope or if it is the symbol of an operator $\cdot_c$ or $^{*_c}$.
An expression (resp. a system $\mathcal{X}$) is said to be \emph{closed} if all of the occurrences of a bounded symbol are bounded.
In this case, the set $\mathrm{free}(\mathcal{X})$ contains the symbols of $\Sigma_0$ that are not bounded.

Let us first show that the closedness is preserved by substitution, factorization and contraction.
\begin{lemma}\label{lem sub pres closed}
  Let $F$ and $F'$ be two closed expressions over $\Sigma,\mathbb{E}$ such that the bounded symbols of $F$ are bounded in $F'$.
  Let $\mathbb{E}_k$ be a variable in $\mathbb{E}$.
  Then:
  \begin{align*}
    F_{\mathbb{E}_k\leftarrow F'}\text{ is closed.}
  \end{align*}
\end{lemma}
\begin{proof}
  By induction over the structure of $F$.
  Let us define for any expression $H$, the expression $G(H)=H_{\mathbb{E}_k\leftarrow F'}$.
  Let us set $G=G(F)$.
  \begin{enumerate}
    \item If $F\in\Sigma_0\cup\{0\}\cup\mathbb{E}\setminus\{\mathbb{E}_k\}$, then $G=F$.
    Therefore $G$ is closed.
    \item If $F=\mathbb{E}_k$, then $G=F'$.
    Therefore $G$ is closed.
    \item If $F=f(E_1,\ldots,E_n)$, then $G=f(G(E_1),\ldots,G(E_n))$.
    By induction hypothesis, $G(E_1)$,$\ldots$, and $G(E_n)$ are closed, and as a consequence so is $G$.
    \item If $F=E_1+E_2$, then $G=G(E_1)+G(E_2)$.
    By induction hypothesis, $G(E_1)$ and $G(E_2)$ are closed, and therefore so is $G$.
    \item If $F=E_1\cdot_c E_2$, then $G=G(E_1)\cdot_c G(E_2)$.
    By induction hypothesis, $G(E_1)$ and $G(E_2)$ are closed.
    Since the bounded symbols of $F$ are bounded in $F'$, $c$ is bounded in $G(E_1)$.
    Consequently, $G$ is closed.
    \item If $F=E_1^{*_c}$, then $G=(G(E_1))^{*_c}$.
    By induction hypothesis, $G(E_1)$ is closed.
    Since the bounded symbols of $F$ are bounded in $F'$, $c$ is bounded in $G(E_1)$.
    Consequently, $G$ is closed.    
  \end{enumerate}
  \qed
\end{proof}

As two direct consequences of Lemma~\ref{lem sub pres closed}:
\begin{corollary}\label{cor sub pres closed}
  Let $\mathcal{X}$ be an equation system over $n$ variables.
  Let $1\leq k\leq n$ be an integer.
  Then:
  \begin{equation*}
    \mathcal{X}^k\text{ is closed.}
  \end{equation*}
\end{corollary}

\begin{corollary}\label{cor forme fact pre clos}
  Let $F$ be a closed expressions over $\Sigma,\mathbb{E}$.
  Let $\mathbb{E}_k$ be a variable in $\mathbb{E}$.
  Let $k-\mathrm{split}(F_n)=(F',F'')$.
  Let $a$ be a symbol not in $\Sigma$.
  Then:
  \begin{align*}
    (F')_{\mathbb{E}_k\leftarrow a}\cdot_a \mathbb{E}_k + F'' \text{ is closed.}
  \end{align*}
\end{corollary}

The stability of the closedness by contraction is even easier to prove; since it is not an inductive transformation:
\begin{lemma}\label{lem contrac pres close}
  Let $E=F\cdot_c F'+F''$ be a closed expression.
  Then:
  \begin{align*}
    F^{*_c}\cdot_c F'' \text{ is a closed expression.}
  \end{align*}
\end{lemma}
\begin{proof}
  Let $E'=F^{*_c}\cdot_c F''$.
  Suppose that $E'$ is not closed.
  Either there exists an occurrence of $c$ that is not bounded in $F''$, or there exists an operator in $\{\cdot_a,^{*_a}\}$ appearing in $F$ (resp. $F''$) such that an occurrence of $a$ is not bounded in $F''$ (resp. in $F$).
  Contradiction with the closedness of $E$.
  \qed
\end{proof}

\begin{corollary}
  Let $\mathcal{X}=\{(\mathbb{E}_j=F_j)\mid 1\leq j\leq n\}$ be a closed equation system.
  Let $1\leq k\leq n$ be an integer such that $\mathbb{E}_k=F'_k\cdot_c \mathbb{E}_k+F''_k$.
  Then:
  \begin{equation*}
    \mathcal{X}_k\text{ is closed.}
  \end{equation*}
\end{corollary}

Finally, let us show that a closed system can be effectively solved: we show that it admits some rational solutions, \emph{i.e.} solutions formed by rational languages. 
And we give a way to compute expressions to denote it.
In the following, we say that a $n$-tuple of rational expressions $(E_1,\ldots,E_n)$ \emph{denotes} a rational solution $(L_1,\ldots, L_n)$ if $L_i=L(E_i)$ for any $1\leq i\leq n$.
The following example illustrates how to compute some rational expressions denoting a solution.

\begin{example}
  Let us consider the closed system $\mathcal{X}$ of Example~\ref{ex syst eq}.
  By substitution of $\mathbb{E}_3$, we obtain
  \begin{align*}
  \mathcal{X}^3&=
    \begin{cases}
      \mathbb{E}_1 & =f(\mathbb{E}_1,\mathbb{E}_1)+f(\mathbb{E}_2,\mathbb{E}_4)\\
      \mathbb{E}_2 & =b+f(\mathbb{E}_2,a+h(\mathbb{E}_4))\\
      \mathbb{E}_3 & =a+h(\mathbb{E}_4)\\
      \mathbb{E}_4 & =a+h(a+h(\mathbb{E}_4))
    \end{cases}
  \end{align*}
  The $4\mathrm{-split}$ of $a+h(a+h(\mathbb{E}_4))$ leads to the factorization $(h(a+h(x_4)))\cdot_{x_4}\mathbb{E}+a$, contracted in $(h(a+h(x_4)))^{*_{x_4}}\cdot_{x_4}a$.
  Then, we obtain
  \begin{align*}
    \begin{cases}
      \mathbb{E}_1 & =f(\mathbb{E}_1,\mathbb{E}_1)+f(\mathbb{E}_2,\mathbb{E}_4)\\
      \mathbb{E}_2 & =b+f(\mathbb{E}_2,a+h(\mathbb{E}_4))\\
      \mathbb{E}_3 & =a+h(\mathbb{E}_4)\\
      \mathbb{E}_4 & =(h(a+h(x_4)))^{*_{x_4}}\cdot_{x_4}a
    \end{cases}
  \end{align*}
  By substitution,
  \begin{align*}
    \begin{cases}
      \mathbb{E}_1 & =f(\mathbb{E}_1,\mathbb{E}_1)+f(\mathbb{E}_2,(h(a+h(x_4)))^{*_{x_4}}\cdot_{x_4}a)\\
      \mathbb{E}_2 & =b+f(\mathbb{E}_2,a+h((h(a+h(x_4)))^{*_{x_4}}\cdot_{x_4}a))\\
      \mathbb{E}_3 & =a+h((h(a+h(x_4)))^{*_{x_4}}\cdot_{x_4}a)\\
      \mathbb{E}_4 & =(h(a+h(x_4)))^{*_{x_4}}\cdot_{x_4}a
    \end{cases}
  \end{align*}
  The $2\mathrm{-split}$ of $b+f(\mathbb{E}_2,a+h((h(a+h(x_4)))^{*_{x_4}}\cdot_{x_4}a))$ leads to the factorization $(f(x_2,a+h((h(a+h(x_4)))^{*_{x_4}}\cdot_{x_4}a)))\cdot_{x_2}\mathbb{E}_2+b$, contracted in $(f(x_2,a+h((h(a+h(x_4)))^{*_{x_4}}\cdot_{x_4}a)))^{*_{x_2}}\cdot_{x_2} b$. Thus, we obtain the new system
  \begin{align*}
    \begin{cases}
      \mathbb{E}_1 & =f(\mathbb{E}_1,\mathbb{E}_1)+f((f(x_2,a+h((h(a+h(x_4)))^{*_{x_4}}\cdot_{x_4}a)))^{*_{x_2}}\cdot_{x_2} b,(h(a+h(x_4)))^{*_{x_4}}\cdot_{x_4}a)\\
      \mathbb{E}_2 & =(f(x_2,a+h((h(a+h(x_4)))^{*_{x_4}}\cdot_{x_4}a)))^{*_{x_2}}\cdot_{x_2} b\\
      \mathbb{E}_3 & =a+h((h(a+h(x_4)))^{*_{x_4}}\cdot_{x_4}a)\\
      \mathbb{E}_4 & =(h(a+h(x_4)))^{*_{x_4}}\cdot_{x_4}a
    \end{cases}
  \end{align*}
  Finally, factorizing/contracting the first equation, we obtain the solution
  \begin{align*}
    \begin{cases}
      \mathbb{E}_1 & =(f(x_1,x_1))^{*_{x_1}}\cdot_{x_1}(f((f(x_2,a+h((h(a+h(x_4)))^{*_{x_4}}\cdot_{x_4}a)))^{*_{x_2}}\cdot_{x_2} b,(h(a+h(x_4)))^{*_{x_4}}\cdot_{x_4}a))\\
      \mathbb{E}_2 & =(f(x_2,a+h((h(a+h(x_4)))^{*_{x_4}}\cdot_{x_4}a)))^{*_{x_2}}\cdot_{x_2} b\\
      \mathbb{E}_3 & =a+h((h(a+h(x_4)))^{*_{x_4}}\cdot_{x_4}a)\\
      \mathbb{E}_4 & =(h(a+h(x_4)))^{*_{x_4}}\cdot_{x_4}a
    \end{cases}
  \end{align*}
\end{example}

Any closed system admits a \emph{canonical} resolution, defined in the proof of the following theorem.
\begin{theorem}\label{thm resol syst clos}
  Let $\mathcal{X}=\{(\mathbb{E}_j=F_j)\mid 1\leq j\leq n\}$ be a closed equation system over a graded alphabet $\Sigma$ and over variables $\{\mathbb{E}_1,\ldots,\mathbb{E}_k\}$.
  Then
  \begin{equation*}
    \mathcal{X}\text{ admits a regular solution over }\mathrm{free}(\mathcal{X}).
  \end{equation*}
  Furthermore, a $n$-tuple of rational expressions denoting this solution can be computed.
\end{theorem}
\begin{proof}
  By recurrence over the cardinal of $\mathcal{X}$.
  \begin{enumerate}
    \item\label{item preuve} Suppose that the equation $\mathbb{E}_n=F_n$ is not recursive.
    \begin{enumerate}
      \item If $n=1$, then $F_1$ is a rational expression and therefore $L(F_1)$ is the unique solution for $\mathcal{X}$.
      Since $\mathcal{X}$ is closed, $L(F_1)\subset T(\mathrm{free}(\mathcal{X}))$.
      \item Otherwise, consider the system $\mathcal{X}'=\mathcal{X}^k\setminus\{\mathbb{E}_n=F_n\}$.
      From Corollary~\ref{cor sub pres closed}, the system $\mathcal{X}'$ is closed.
      By recurrence hypothesis, $\mathcal{X}'$ admits a regular solution $Z=(L_1,\ldots,L_{n-1})$ over $\mathrm{free}(\mathcal{X})$ denoted by $(E_1,\ldots,E_{n-1})$. 
      From Lemma~\ref{lem sub cas eq pas rec}, this implies that $(L_1,\ldots,L_{n-1},L_Z(F_n))$ is a solution for $\mathcal{X}$ that is, by construction of $Z$, a solution over $\mathrm{free}(\mathcal{X})$.
      From Lemma~\ref{lem sub cons lang}, $L_Z(F_n)$ is denoted by $E_n=(\ldots(F_n)_{\mathbb{E}_1\leftarrow E_1}\ldots)_{_{\mathbb{E}_{n-1}\leftarrow E_{n-1}}}$, that is a rational expression with no variables.
      Therefore $\mathcal{X}$ admits a regular solution $(L_1,\ldots,L_{n-1},L_Z(F_n))$ over $\mathrm{free}(\mathcal{X})$ denoted by $(E_1,\ldots,E_{n})$. 
    \end{enumerate}
    \item Consider that the equation $\mathbb{E}_n=F_n$ is recursive.
    Let $k\mathrm{split}(F_n)=(F',F'')$.
    Let $a$ be a symbol not in $\Sigma$.
    Let $F'_n=(F')_{\mathbb{E}_k\leftarrow a}\cdot_a \mathbb{E}_k + F''$.
    Since $\mathcal{X}$ is closed, it holds from Proposition~\ref{prop sub par split ok} that $\mathcal{X}$ admits a solution over $\mathrm{free}(\mathcal{X})$ if and only if $\mathcal{X}'=(\mathcal{X}\setminus\{\mathbb{E}_n=F_n\})\cup\{\mathbb{E}_n=F'_n\}$ does.
    From Corollary~\ref{cor forme fact pre clos}, $\mathcal{X}'$ is closed.
    From Proposition~\ref{prop sol pour eq rec}, $\mathcal{X}'$ admits a solution over $\mathrm{free}(\mathcal{X})$ if and only if $\mathcal{X}'_n$ does.
    From Lemma~\ref{lem contrac pres close}, $\mathcal{X}'_n$ is closed, and contains the equation $\mathbb{E}_k=F'^{*_c}\cdot_c F''$, that is not recursive.
    The existence of the solution is then proved from the point~(\ref{item preuve}).
  \end{enumerate}
  \qed
\end{proof}

In other words,
\begin{theorem}\label{thm closed syst solv}
  Any closed equation system is effectively solvable.
\end{theorem}

\section{Construction of a Rational Tree Expression from an Automaton}\label{co}

In this section, we show how to extract a tree languages equations system from a given FTA $\mathcal{A}=(\Sigma,Q,Q_f,\Delta)$.
Then, using the Arden's Lemma and the transformations (contraction and substitution) defined in the previous sections, we show how to resolve it and compute an equivalent rational expression $E_q$ by associating with a state $q$ in $Q$ an equation defining $L(q)$.
Let us first recall a basic property of the down language of a state:

\begin{lemma}\label{lem decompo down lang}
  Let $\mathcal{A}=(\Sigma,Q,Q_f,\Delta)$ be a FTA.
  Let $q\in Q$ be a state.
  Then:
  \begin{equation*}
    L(q)= \bigcup_{(f,q_1,\ldots,q_n,q)\in \Delta}f(L(q_1),\ldots,L(q_n))
  \end{equation*}
\end{lemma}
\begin{proof}
  Let us set $L'(q)=\bigcup_{(f,q_1,\ldots,q_n,q)\in \Delta}f(L(q_1),\ldots,L(q_n))$.
  Let $t=f(t_1,\ldots,t_n)$ be a tree in $T(\Sigma)$.
  Let us show that $t\in L(q)$ $\Leftrightarrow$ $t\in L'(q)$.
  By definition, $t\in L(q)   \Leftrightarrow q\in \delta(t)$.
  Then:
  \begin{align*}
    q\in \delta(t) & \Leftrightarrow \exists (f,q_1,\ldots,q_n,q)\in \Delta, (\forall 1\leq i\leq n,q_i\in\delta(t_i))\\
    & \Leftrightarrow \exists (f,q_1,\ldots,q_n,q)\in \Delta, (\forall 1\leq i\leq n,t_i\in L(q_i))\\
    & \Leftrightarrow \exists (f,q_1,\ldots,q_n,q)\in \Delta, t\in f(L(q_1),\ldots,L(q_n))\\
    & \Leftrightarrow t\in L'(q)
  \end{align*}
  \qed
\end{proof}

The previous lemma can be used to define an equation system that can describe the relations between the down languages of the states of a given FTA.

Let $\mathcal{A}=(\Sigma,Q,Q_f,\Delta)$ be a FTA with $Q=\{1,\ldots,n\}$.
The \emph{equation system} associated with $\mathcal{A}$ is the set of equations $\mathcal{X}_{\mathcal{A}}$ over the variables $\mathbb{E}_1,\ldots,\mathbb{E}_n$ defined by $\mathcal{X}_{\mathcal{A}}=\{\mathcal{E}_q \mid q\in Q\}$ where for any state $q$ in $Q$,  $\mathcal{E}_q$ is the equation $\mathbb{E}_q=F_q$ with $F_q= \sum_{(f,q_1,\ldots,q_n,q)\in \Delta} f(\mathbb{E}_{q_1},\ldots,\mathbb{E}_{q_n})$.
Let us show that any solution of $\mathcal{X}_{\mathcal{A}}$ denotes the down languages of the states of $\mathcal{A}$.

\begin{proposition}\label{prop lang bas ds syst}
  Let $\mathcal{A}=(\Sigma,Q,Q_f,\Delta)$ be a FTA with $Q=\{1,\ldots,n\}$.
  Let $\mathbb{E}=(E_1,\ldots,E_n)$ be a solution of $\mathcal{X}_{\mathcal{A}}$.
  Then:
  \begin{equation*}
    \forall 1\leq j\leq n, L(E_j)=L(j).
  \end{equation*}
\end{proposition}
\begin{proof}
  Let $t$ be tree over $\Sigma$.
  Let us show by induction over $t$ that $t\in L(E_j)$ $\Leftrightarrow$ $t\in L_j$.
  \begin{enumerate}
    \item Consider that $t\in\Sigma_0$.
    Then
    \begin{align*}
      t\in L(E_j) & \Leftrightarrow t\in L(\sum_{(f,q_1,\ldots,q_n,j)\in\Delta} f(E_{q_1},\ldots,E_{q_n}))\\
      & \Leftrightarrow (t,j)\in\Delta\\
      & \Leftrightarrow t\in L(j)
    \end{align*}
    \item Otherwise, $t=g(t_1,\ldots,t_k)$ and
    \begin{align*}
      t\in L(E_j) & \Leftrightarrow t\in L(\sum_{(f,q_1,\ldots,q_n,j)\in\Delta} f(E_{q_1},\ldots,E_{q_n}))\\
      & \Leftrightarrow \exists (g,q_1,\ldots,q_k,j)\in\Delta\wedge \forall 1\leq l\leq k, t_l\in L(E_{q_l})\\
      & \Leftrightarrow \exists (g,q_1,\ldots,q_k,j)\in\Delta\wedge \forall 1\leq l\leq k, t_l\in L(q_l) & \text{ (induction hypothesis)}\\
      & \Leftrightarrow t\in \bigcup_{(f,q_1,\ldots,q_n,j)\in\Delta} f(L(q_1),\ldots,L(q_n))\\
      & \Leftrightarrow t\in L(j)& \text{(Lemma~\ref{lem decompo down lang})}
    \end{align*}
  \end{enumerate}
  \qed
\end{proof}

Since $\mathcal{X}_A$ is by definition closed, it holds from Theorem~\ref{thm closed syst solv} that

\begin{theorem}\label{thm syst aut solv}
  Let $\mathcal{A}=(\Sigma,Q,Q_f,\Delta)$ be a FTA.
  Then:
  \begin{equation*}
    \mathcal{X}_{\mathcal{A}}\text{ can be effectively solved.}
  \end{equation*} 
\end{theorem}

As a direct consequence of Theorem~\ref{thm syst aut solv} and of Proposition~\ref{prop lang bas ds syst}, following Equation~\eqref{eq lang et lang bas},

\begin{theorem}
  Let $\mathcal{A}=(\Sigma,\{1,\ldots,n\},Q_f,\Delta)$ be a FTA.
  Let $(E_1,\ldots,E_n)$ denoting a solution of $\mathcal{X}_A$.
  Then:
  \begin{align*}
    L(\mathcal{A}) \text{ is denoted by the rational expression }\sum_{j\in Q_f} E_j.
  \end{align*}
\end{theorem}

\begin{example}
  Let us consider the FTA $A$ in Figure~\ref{fig ex aut}.
  The system associated with $A$ is the system $\mathcal{X}$ in Example~\ref{ex syst eq}:
  \begin{align*}
  \mathcal{X}&=
    \begin{cases}
      \mathbb{E}_1 & =f(\mathbb{E}_1,\mathbb{E}_1)+f(\mathbb{E}_2,\mathbb{E}_4)\\
      \mathbb{E}_2 & =b+f(\mathbb{E}_2,\mathbb{E}_4)\\
      \mathbb{E}_3 & =a+h(\mathbb{E}_4)\\
      \mathbb{E}_4 & =a+h(\mathbb{E}_3)
    \end{cases}
  \end{align*}
  Let us apply the resolution defined in the proof of Theorem~\ref{thm resol syst clos}.
  We first compute $\mathcal{X}_4$:
  \begin{align*}
  \mathcal{X}^4&=
    \begin{cases}
      \mathbb{E}_1 & =f(\mathbb{E}_1,\mathbb{E}_1)+f(\mathbb{E}_2,a+h(\mathbb{E}_3))\\
      \mathbb{E}_2 & =b+f(\mathbb{E}_2,a+h(\mathbb{E}_3))\\
      \mathbb{E}_3 & =a+h(a+h(\mathbb{E}_3))\\
      \mathbb{E}_4 & =a+h(\mathbb{E}_3)
    \end{cases}
  \end{align*}
  Then we have to solve the closed subsystem
  \begin{align}
    \begin{cases}
      \mathbb{E}_1 & =f(\mathbb{E}_1,\mathbb{E}_1)+f(\mathbb{E}_2,a+h(\mathbb{E}_3))\\
      \mathbb{E}_2 & =b+f(\mathbb{E}_2,a+h(\mathbb{E}_3))\\
      \mathbb{E}_3 & =a+h(a+h(\mathbb{E}_3))
    \end{cases}\label{eq syst}
  \end{align}
  The $3-\mathrm{split}$ of $a+h(a+h(\mathbb{E}_3))$ leads to the factorization $h(a+h(x_3))\cdot_{x_3} \mathbb{E}_3+a$, contracted in $(h(a+h(x_3)))^{*_{x_3}}\cdot_{x_3}a$.
  Thus, the system~\eqref{eq syst} is equivalent to
  \begin{align*}
    \begin{cases}
      \mathbb{E}_1 & =f(\mathbb{E}_1,\mathbb{E}_1)+f(\mathbb{E}_2,a+h(\mathbb{E}_3))\\
      \mathbb{E}_2 & =b+f(\mathbb{E}_2,a+h(\mathbb{E}_3))\\
      \mathbb{E}_3 & =(h(a+h(x_3)))^{*_{x_3}}\cdot_{x_3}a
    \end{cases}
  \end{align*}
  and by substitution of $\mathbb{E}_3$ to
  \begin{align*}
    \begin{cases}
      \mathbb{E}_1 & =f(\mathbb{E}_1,\mathbb{E}_1)+f(\mathbb{E}_2,a+h((h(a+h(x_3)))^{*_{x_3}}\cdot_{x_3}a))\\
      \mathbb{E}_2 & =b+f(\mathbb{E}_2,a+h((h(a+h(x_3)))^{*_{x_3}}\cdot_{x_3}a))\\
      \mathbb{E}_3 & =(h(a+h(x_3)))^{*_{x_3}}\cdot_{x_3}a
    \end{cases}
  \end{align*}
  Now, let us solve the new subsystem
  \begin{align}
    \begin{cases}
      \mathbb{E}_1 & =f(\mathbb{E}_1,\mathbb{E}_1)+f(\mathbb{E}_2,a+h((h(a+h(x_3)))^{*_{x_3}}\cdot_{x_3}a))\\
      \mathbb{E}_2 & =b+f(\mathbb{E}_2,a+h((h(a+h(x_3)))^{*_{x_3}}\cdot_{x_3}a))
    \end{cases}\label{eq syst 2}
  \end{align}
  The $2\mathrm{-split}$ of $b+f(\mathbb{E}_2,a+h((h(a+h(x_3)))^{*_{x_3}}\cdot_{x_3}a))$ leads to the factorization $(f(x_2,a+h((h(a+h(x_3)))^{*_{x_3}}\cdot_{x_3}a)))\cdot_{x_2}\mathbb{E}_2+b$, contracted in $((f(x_2,a+h((h(a+h(x_3)))^{*_{x_3}}\cdot_{x_3}a))))^{*_{x_2}}\cdot_{x_2} b$.
  Consequently, the system~\eqref{eq syst 2} is equivalent to
  \begin{align*}
    \begin{cases}
      \mathbb{E}_1 & =f(\mathbb{E}_1,\mathbb{E}_1)+f(\mathbb{E}_2,a+h((h(a+h(x_3)))^{*_{x_3}}\cdot_{x_3}a))\\
      \mathbb{E}_2 & =((f(x_2,a+h((h(a+h(x_3)))^{*_{x_3}}\cdot_{x_3}a))))^{*_{x_2}}\cdot_{x_2} b
    \end{cases}
  \end{align*}
  and by substitution to
  \begin{align*}
    \begin{cases}
      \mathbb{E}_1 & =f(\mathbb{E}_1,\mathbb{E}_1)+f(((f(x_2,a+h((h(a+h(x_3)))^{*_{x_3}}\cdot_{x_3}a))))^{*_{x_2}}\cdot_{x_2} b,a+h((h(a+h(x_3)))^{*_{x_3}}\cdot_{x_3}a))\\
      \mathbb{E}_2 & =((f(x_2,a+h((h(a+h(x_3)))^{*_{x_3}}\cdot_{x_3}a))))^{*_{x_2}}\cdot_{x_2} b
    \end{cases}
  \end{align*}
  Then, by factorization/contraction,
  \begin{align*}
      \mathbb{E}_1 & =(f(x_1,x_1))^{*_{x_1}}\cdot_{x_1} (f(((f(x_2,a+h((h(a+h(x_3)))^{*_{x_3}}\cdot_{x_3}a))))^{*_{x_2}}\cdot_{x_2} b, a+h((h(a+h(x_3)))^{*_{x_3}}\cdot_{x_3}a)))
  \end{align*}
  Finally, we obtain the solution
  \begin{align*}
    \begin{cases}
      \mathbb{E}_1 & =(f(x_1,x_1))^{*_{x_1}}\cdot_{x_1} (f(((f(x_2,a+h((h(a+h(x_3)))^{*_{x_3}}\cdot_{x_3}a))))^{*_{x_2}}\cdot_{x_2} b,a+h((h(a+h(x_3)))^{*_{x_3}}\cdot_{x_3}a)))\\
      \mathbb{E}_2 & =((f(x_2,a+h((h(a+h(x_3)))^{*_{x_3}}\cdot_{x_3}a))))^{*_{x_2}}\cdot_{x_2} b\\
      \mathbb{E}_3 & =(h(a+h(x_3)))^{*_{x_3}}\cdot_{x_3}a\\
      \mathbb{E}_4 & =a+h((h(a+h(x_3)))^{*_{x_3}}\cdot_{x_3}a)
    \end{cases}
  \end{align*}
  Since the final states are $1$ and $3$, it holds that $L(\mathcal{A})$ is denoted by:
  \begin{align*}
   \begin{split}
    &(f(x_1,x_1))^{*_{x_1}}\cdot_{x_1} (f(((f(x_2,a+h((h(a+h(x_3)))^{*_{x_3}}\cdot_{x_3}a))))^{*_{x_2}}\cdot_{x_2} b,a+h((h(a+h(x_3)))^{*_{x_3}}\cdot_{x_3}a)))\\
    &\quad +(h(a+h(x_3)))^{*_{x_3}}\cdot_{x_3}a
   \end{split}
  \end{align*}
\end{example}

  \begin{figure}[H]
    \centerline{
    \begin{tikzpicture}[node distance=2.5cm,bend angle=30,transform shape,scale=1]
      \node[state,accepting] (q1) {$1$};
      \node[state, below left of=q1] (q2) {$2$};
      \node[state, below right of=q1] (q4) {$4$};
      \node[state, right of=q4,accepting] (q3) {$3$};
        \draw (q2) ++(-1cm,0cm) node {$b$}  edge[->] (q2); 
        \draw (q3) ++(1cm,0cm) node {$a$}  edge[->] (q3);  
        \draw (q4) ++(0cm,-1cm) node {$a$}  edge[->] (q4); 
        \path[->]
          (q3) edge[->,below,bend left] node {$h$} (q4)
          (q4) edge[->,above,bend left] node {$h$} (q3)
        ;
        \draw (q2) ++(0.5cm,1cm)  edge[->,in=135,out=135,looseness=2] node[above right,pos=0] {$f$} (q2) edge[shorten >=0pt] (q2) edge[shorten >=0pt] (q4);       
        \draw (q4) ++(-0.5cm,1cm)  edge[->] node[above right,pos=0] {$f$} (q1) edge[shorten >=0pt] (q2) edge[shorten >=0pt] (q4);      
        \draw (q1) ++(0.5cm,0.5cm)  edge[->,in=90,out=45,looseness=3] node[right,pos=0] {$f$} (q1) edge[shorten >=0pt,bend left] (q1) edge[shorten >=0pt,bend right] (q1); 
      \end{tikzpicture}
    }
    \caption{The FTA $A$.}
    \label{fig ex aut}
  \end{figure}

\section{Conclusion}\label{con}
We present a new construction of a rational expression from a tree automaton. 
This construction, based on a generalization of Arden's Lemma, gives another way to prove Kleene's theorem for tree.
In order to produce the expression, we studied the notion of tree languages equation systems and determine a sufficient condition to solve them.
The next step is to study the different links that may exist between the different methods of computation of an expression from an automaton, like it was studied in~\cite{Sak05}.

\bibliographystyle{plain}
\bibliography{biblio}

\end{document}